\documentclass{article}
\usepackage[english]{babel}
\usepackage[letterpaper,top=2cm,bottom=2cm,left=2cm,right=2cm,marginparwidth=1.75cm]{geometry}
\usepackage{tgheros}
\usepackage[T1]{fontenc}
\usepackage{amsmath}
\usepackage{amsfonts}
\usepackage{amssymb}
\usepackage{amsthm}
\usepackage{mathtools}
\usepackage[colorlinks=true, allcolors=blue]{hyperref}
\usepackage{asymptote}
\usepackage{changepage}
\usepackage{graphicx}
\usepackage{url}
\usepackage[labelfont=bf]{caption}
\usepackage{appendix}

\usepackage[style=alphabetic, maxnames=10, url=false]{biblatex}
\usepackage{csquotes}

\newtheorem{theorem}{Theorem}[section]
\newtheorem{lemma}[theorem]{Lemma}

\newtheorem{definition}[theorem]{Definition}

\newcommand{\R}{\mathbb{R}}
\newcommand{\RR}{\mathbb{R}}

\DeclareMathOperator{\tr}{tr}

\DeclareMathOperator{\nnz}{nnz}
\DeclareMathOperator{\poly}{poly}
\newcommand{\OPT}{\textnormal{OPT}}

\title{Faster high-accuracy multicommodity flow in dense graphs}
\author{Chenxin Dai, Alicia Stepin, Colin Tang\footnote{Carnegie Mellon University, $\{\text{chenxin2,astepin,cstang}\}\text{@andrew.cmu.edu}$}}

\begin{document}
\maketitle

\begin{abstract}
    We give a fast algorithm for solving min-cost $k$-commodity flow. The basic idea is to construct an auxiliary linear program that has low rank and whose minimum value is at most $1/k$ times the minimum value of the original problem (thus, solving this auxiliary linear program will make at least $1/k$ fraction of progress in the original problem). Low-rank linear programs can be solved quickly using black-box techniques. Thus, our algorithm runs in time $\tilde{O}(\poly(k)(n^{2.5}+m\sqrt{n}))$ on a directed graph with $n$ vertices and $m$ edges. We do not rely on any fast matrix multiplication.
\end{abstract}

\section{Introduction}

Linear programming remains one of the most fundamental primitives in convex optimization.
Current research is devoted to improving interior-point methods for solving linear programs. See \cite{kv96, LS15a, chen2022maximumflowminimumcostflow} for previous results. A single 
interior-point iteration in these methods typically reduces the problem of making progress along a central path to solving a few auxiliary linear systems. So, the complexity of the algorithm is governed both by the number of iterations and by the cost of these linear solves, which makes structures classes of linear programs particularly amenable to faster algorithms. 

This perspective has been especially fruitful for flow problems, which, themselves, have been studied extensively. The class of network-flow problem is a classic class of problems in computer science, where very efficient algorithms are known. Single-commodity max flow and min-cost flow are fundamental problems in optimization, where their study begins with the classic augmenting-path result of Fork and Fulkerson \cite{FordFulkerson1956}. More recently, interior point methods for these problems have pushed the runtime to nearly linear (in the number of edges) time, first by a randomized algorithm and subsequently deterministically \cite{chen2022maximumflowminimumcostflow, BrandEtAl2023Deterministic}. Recent work has also produced a randomized parallel algorithm for exact single-commodity min-cost flow with $\widetilde O(m+n^{3/2})$ work and $\widetilde O(\sqrt n)$ depth \cite{BrandGholizadehJiangVos2025}.

We focus this paper on the problem of \textbf{minimum-cost $k$-commodity flow}.

%However, if we generalize to multicommodity flow, in some settings this is known to be as hard as solving general linear programs (\cite{ding2022twocommodityflowequivalentlinear}). One motivation for studying flow problems comes from empirical risk minimization.
Take $G=(V,E)$ be a directed graph with $n$ vertices and $m$ edges, $k$ a positive integer,  $c\in\RR^m$, and  $d_1,d_2,\dots,d_k\in \RR^n$. Given some cost vector $c \in \RR^m$ and demand vectors $d_1,\ldots,d_k\in\R^n$, where $k\ge2$, write
$[k]=\{1,\ldots,k\}$. Thus, the minimum-cost $k$-commodity flow problem is the following linear program

\begin{equation}\label{eq:min-cost-flow}\begin{array}{rll}
    \text{Minimize}&c^\top (f_1+f_2+\dots+f_k)&\\
    \text{where}&f_1,f_2,\dots,f_k\in\RR^m &\text{are variables}\\
    \text{subject to:}& f_i(e)\ge 0 &\text{for all } 1\le i \le k \text{ and } e \in E\\
    &f_1(e)+f_2(e)+\dots+f_k(e)\le 1&\text{for all }e\in E\\
    &B^\top f_i = d_i &\text{for all } 1 \le i \le k
\end{array}\end{equation}
Here $B\in\RR^{m\times n}$ is the edge-vertex directed incidence matrix of $G$: for any directed edge $e = (v,w)\in E$, we have $B_{ev} = -1$ and $B_{ew}=+1$ with all zeros elsewhere.

Think of the variables $f_1,f_2,\dots,f_k$ as representing \textbf{flows} of $k$ distinct commodities, the vector $c$ as representing the \textbf{cost} of using each edge, and the vectors $d_1,d_2,\dots,d_k$ as representing the desired \textbf{demand} on each vertex (the excess or deficit induced by the flow of a commodity towards or away from a vertex). The constraint $f_1(e)+f_2(e)+\dots+f_k(e)\le 1$ says that the combined flow of all $k$ commodities cannot exceed the \textbf{capacity} of each edge, which is taken to be $1$ here.

For high-accuracy minimum-cost multicommodity flow on possibly
directed graphs, the strongest specialized running-time bound is due to van den Brand and Zhang~\cite{vandenBrandZhang2023}. For integer capcities, costs and demands of polynomial magnitude, and for $\epsilon^{-1}\le\poly(n)$, their deterministic algorithm runs in time $\widetilde O\!\left(k^{5/2}\sqrt m\,n^{\omega-1/2}\right)$. It performs 
$\widetilde O(\sqrt{km}\log(CU/\epsilon))$ sequential interior-point
iterations, where $C, U$ bound the numerical data. The returned flows have cost at most $\OPT + \varepsilon$, are nonnegative and obey shared capacities exactly, while conservation is guaranteed up to an additive $\epsilon$ error.
We remark that the shared cost vector and unit capacities are a special case of this formulation. 
We note that this theorem is the closest result when it comes to total-running-time comparison, but it is neither an exact-conservation nor a parallel-depth.

A recent focus in this area has been to present an algorithm with a runtime dependent on the \emph{intrinsic dimension} of the problem, which can be roughly thought of as the rank.
We remark on the rank-sensitive box-LP method of ~\cite{LS20}. They introduce one unused-capacity variable per edge for a total of $(k+1)m$ variables and equality rank exactly
$m+k(n-q)$, and so it requires
$\widetilde O(\sqrt{m+k(n-q)}\log(1/\epsilon))$ adaptive weighted
normal-system rounds.
The goal of this paper is to reduce the adaptive solve count in the dense regime. Specifically, we remove the $m$ shared-capacity equations from the linear systems without relaxing shared capacities. We do this by encoding feasible edgewise transfers through our zonotope method.

We make the following mild assumptions: \begin{equation}\label{eq:flow-assumptions}\begin{aligned}
    n &\le m \le \poly(n)\\
    k &\le n^{0.1}\\
    \|c\|_\infty, \|d_1\|_\infty, \|d_2\|_\infty, \dots, \|d_k\|_\infty &\le \poly(n)
\end{aligned}\end{equation}

We present an algorithm for solving \eqref{eq:min-cost-flow} that runs in time $\tilde{O}(\poly(k)(n^{2.5}+m\sqrt{n}))$.%, where $\omega < 2.38$ is the matrix multiplication exponent. Note that $2.38<2.5$. 

%\begin{theorem}\label{thm:slow-flow}
%    There exists an algorithm $\textnormal{\texttt{SlowAlgorithm}}$ which, assuming the mild assumptions in \eqref{eq:flow-assumptions}, can obtain a solution to \eqref{eq:min-cost-flow} that attains value at most $1/\poly(n)$ plus the minimum value of \eqref{eq:min-cost-flow}. The algorithm runs in time $\tilde{O}(\poly(k)n^{2.5})$.
%\end{theorem}

\begin{theorem}\label{thm:slow-flow}
    There exists an algorithm which, assuming the mild assumptions in \eqref{eq:flow-assumptions}, can obtain an almost-feasible solution to \eqref{eq:min-cost-flow} that attains value at most $1/\poly(n)$ plus the minimum value of \eqref{eq:min-cost-flow}. The algorithm runs in time $\tilde{O}(\poly(k)(n^{2.5}+m\sqrt{n}))$. Here, ``almost-feasible'' means that $\|B^\top f_i -d_i\|_\infty \le 1/\poly(n)$.
\end{theorem}

A naive approach to solving the linear program in \eqref{eq:min-cost-flow} might suggest using rank-sensitive linear programming algorithms, but this proves to be expensive because this linear program has high rank. Instead, our algorithm repeatedly constructs and solves subproblems that have low rank, which can be solved quickly. These subproblems are \textbf{zonotope} flows and will be defined in \eqref{eq:zonotope}.

We remark that we cannot hope to find a feasible solution, as opposed to an \emph{almost-}feasible solution, while preserving the runtime. Already for two commodities, exact directed flow has the same polynomial-time expressive power of general linear-programming via a reduction in \cite{Itai78}. For integer-coefficient LP's with polynomially bouned data, \cite{ding2022twocommodityflowequivalentlinear} show nearly-linear-size reductions that apply to high-accuracy approximations with only a polynomial loss in additive error. However, these reductions produce only sparse instances. We also remark that in general the solution to such a system can be rational numbers whose numerator and denominator are $\Omega(n)$ digits, which incurs a $\Omega(n)$-factor blowup of any arithmetic operations.

\paragraph{Notation.}
\begin{center}
\small
\renewcommand{\arraystretch}{1.15}
\begin{tabular}{@{}p{0.19\linewidth}p{0.74\linewidth}@{}}
\hline
Notation & Meaning \\
\hline
$G=(V,E)$ & Directed graph with vertex set $V$ and edge set $E$. \\
$n=|V|$, $m=|E|$ & Numbers of vertices and directed edges, respectively. \\
$k$ and $[k]$ & Number of commodities and the index set
  $[k]=\{1,2,\ldots,k\}$. \\
$B\in\R^{m\times n}$ & Directed edge-vertex incidence matrix. \\
$f_i,c_i\in\R^m$ & Flow and edge-cost vectors for commodity $i$. \\
$d_i\in\R^n$ & Net-demand vector for commodity $i$; conservation is
  $B^\top f_i=d_i$. \\
%$u\in\R_{\geq0}^m$ & Shared edge-capacity vector. \\
$\OPT$ and $\epsilon$ & Optimal objective value and requested additive
  accuracy. \\
$\|x\|_p$ & The $\ell_p$ norm of $x$. \\
$\tilde {O}(\bullet)$ & Asymptotic notation suppressing
  polylogarithmic factors. \\
$\omega$ & The matrix multiplication exponent.\\
\hline
\end{tabular}
\end{center}

\section{Zonotope flows}

Let us briefly describe our algorithm for Theorem \ref{thm:slow-flow}. Suppose we currently have values for the variables \\ $f_1,f_2,\dots,f_k$, such that $c^\top (f_1+f_2+\dots+f_k)$ is rather close to the minimum. Let us rewrite \eqref{eq:min-cost-flow} as follows: 
\begin{equation}\label{eq:min-cost-flow-delta-f}\begin{array}{rll}
    \text{Minimize}&c^\top (\Delta f_1+\Delta f_2+\dots+\Delta f_k)&\\
    \text{where}&\Delta f_1,\Delta f_2,\dots,\Delta f_k\in\RR^m &\text{are variables}\\
    \text{subject to:}& \Delta f_i(e)\ge -f_i(e) &\text{for all } 1\le i \le k \text{ and } e \in E\\
    &\Delta f_1(e)+\Delta f_2(e)+\dots+\Delta f_k(e)\le 1-f_1(e)-f_2(e)-\dots-f_k(e)&\text{for all }e\in E\\
    &B^\top \Delta f_i = 0 &\text{for all } 1 \le i \le k
\end{array}\end{equation} This is indeed an equivalent problem, since we may update $f_i \gets f_i + \Delta f_i$ (as suggested by the $\Delta$ notation).

The basic idea is to construct a new linear program (which we will call \textbf{zonotope} flow), which has low rank, and whose minimum is at most $\frac{1}{k}$ times the minimum of \eqref{eq:min-cost-flow-delta-f} (note that this is a negative number). By using Lemma \ref{lem:n25} to solve the zonotope linear program quickly, we make at least a $\frac{1}{k}$ amount of progress in the original problem \eqref{eq:min-cost-flow-delta-f}. (Thus, we only need to solve $O(k\log n)$ zonotope linear programs in order to solve \eqref{eq:min-cost-flow} up to $1/\poly(n)$ additive error, as required in Theorem \ref{thm:slow-flow}.)

Let us now describe the zonotope linear program. We are searching for an integer $r \ge k$ and real numbers \\  $(\alpha_{ij})_{1\le i \le k; 1 \le j \le r}$ and vectors $l_1,l_2,\dots,l_r\in\RR^m$ and $u_1,u_2,\dots,u_r \in\RR^m$ and a constant $0<\lambda\le 1$ with a desirable property which will be specified later. 

\begin{definition}
\label{def:zflow}
For a $k$-commodity flow instance as defined at the start of the paper,
an $r$-commodity \textbf{zonotope flow} is given by
mapping coefficients $\alpha \in \R^{k \times r}$
and capacities $l_1, \ldots, l_{r}, u_{1}, \ldots, u_r \in \R^{m}$ and is the following linear program:
\begin{equation}\label{eq:zonotope}
    \begin{array}{rll}
        \text{Minimize} & \displaystyle c^\top \sum_{i=1}^k \sum_{j=1}^r \alpha_{ij} x_j& \\
        \text{where} & x_1,x_2,\dots,x_r\in\RR^m&\text{are variables}\\
        \text{subject to:} &l_j(e) \le x_j(e) \le u_j(e) &\text{for all }1 \le j \le r \text{ and } e \in E \\
        & \displaystyle B^\top \sum_{j=1}^r \alpha_{ij} x_j = 0 &\text{for all } 1 \le i \le k
    \end{array}
\end{equation}
That is, the vertex residue conditions exist in the $k$ commodities mapped to (via $\alpha x$),
but the capacities exist in the $r$ commodities.
\end{definition}

Observe that \eqref{eq:zonotope} is in the form required by Lemma \ref{lem:n25}. Given a feasible point $(x_1,x_2,\dots,x_r)$ of \eqref{eq:zonotope}, we would like to set $\Delta f_i \gets \displaystyle \sum_{j=1}^r \alpha_{ij} x_j$, and hope the following desirable properties hold:
\begin{itemize}
    \item (Feasibility) $(\Delta f_1, \Delta f_2,\dots,\Delta f_k)$ is a feasible point of \eqref{eq:min-cost-flow-delta-f}
    \item ($\lambda$-Efficacy) The minimum value of \eqref{eq:zonotope} is at most $\lambda$ times the minimum value of \eqref{eq:min-cost-flow-delta-f}
\end{itemize} 
%We would also like the following desirable property:\begin{itemize}
%    \item ($\lambda$-Efficacy) The minimum value of \eqref{eq:zonotope} is at most $\lambda$ times the minimum value of \eqref{eq:min-cost-flow-delta-f}
%\end{itemize}
If both Feasibility and $\lambda$-Efficacy hold, we can do the following: Use Lemma \ref{lem:n25} to solve \eqref{eq:zonotope}, which will make at least a $\lambda$ amount of progress in the original problem \eqref{eq:min-cost-flow-delta-f}. We will show in Section \ref{sec:details} that we can find data which satisfy both Feasibility and $\lambda$-Efficacy, with $\lambda = \frac{1}{k}$.

Now let us explain the name ``zonotope''. We are searching for sufficient conditions for Feasibility and $\lambda$-Efficacy. For fixed $e\in E$, define the convex body \begin{multline*}K(e) \coloneq \left\{\left(\sum_{j=1}^r \alpha_{1j} x_j(e),\sum_{j=1}^r \alpha_{2j} x_j(e),\dots,\sum_{j=1}^r \alpha_{kj} x_j(e)\right):x_1(e),x_2(e),\dots,x_r(e)\in\RR \right.\\ \text{ such that } l_j(e)\le x_j(e) \le u_j(e) \text{ for all }1\le j \le r\Bigg\}\end{multline*} This is a linear image of the hyperbox \[\bigtimes_{j=1}^r [l_j(e),u_j(e)]\] and is thus a zonotope, justifying the name. Additionally, define the convex body \begin{multline*}S(e)\coloneq \{(y_1,y_2,\dots,y_k)\in\RR^k: y_i \ge -f_i(e) \text{ for all } 1 \le i \le k,\\ \text{and } y_1+y_2+\dots+y_k \le 1 - f_1(e) - f_2(e) -\dots - f_k(e)\}\end{multline*} Observe that this is a simplex. Now, a sufficient (and necessary) condition for Feasibility is that \begin{equation}\label{eq:feas}K(e)\subseteq S(e)\text{ for all } e\in E,\end{equation} whereas a sufficient condition for $\lambda$-Efficacy is that \begin{equation}\label{eq:effic}\lambda \cdot S(e) \subseteq K(e)\text{ for all } e\in E,\end{equation} where $\lambda \cdot S(e)$ denotes the dilation by a factor of $\lambda$ with respect to the origin. (Indeed, if $(\Delta f_1^*, \Delta f_2^*, \dots, \Delta f_k^*)$ denotes the minimizer of \eqref{eq:min-cost-flow-delta-f}, then $\lambda \cdot (\Delta f_1^*, \Delta f_2^*, \dots, \Delta f_k^*)$ has value $\lambda$ times the minimum of \eqref{eq:min-cost-flow-delta-f}, and is moreover in the image of the mapping $(x_j)\mapsto \left(\displaystyle \sum_{j=1}^r \alpha_{ij} x_j\right)$.)

%It is worth noting that unlike Lee-Sidford 2019, our method requires each entry of $c$ (i.e. all costs) to be nonnegative.

%\section{Notation}
%We write $\RR_{\ge 0}$ for the set of nonnegative real numbers. 

\section{Construction of zonotope flows, and convergence}\label{sec:details}
We give an explicit construction for the data $r,(\alpha_{ij}),l_1,l_2,\dots,l_r,u_1,u_2,\dots,u_r,\lambda$, with $\lambda = \frac{1}{k}$, such that both Feasibility and $\lambda$-Efficacy are satisfied. As observed previously, we only need to verify \eqref{eq:feas} and \eqref{eq:effic}.

Let us slightly reformulate the problem. Write $\Delta^k$ for the $k$-simplex, embedded in $\RR^{k+1}$ as the convex hull of the standard basis vectors $e_1,e_2,\dots,e_{k+1}\in\RR^{k+1}$ (think of this as barycentric coordinates). We will show that for each point $W=(w_1,w_2,\dots,w_{k+1})\in\Delta^k$, there exists a zonotope $K(W)$ satisfying \begin{equation}\label{eq:zonotope-bary}(1-\lambda)W + \lambda \Delta^k \subseteq K(W) \subseteq \Delta^k.\end{equation}

Construct $K(W)$ as follows. For all $1\le \mu < \nu \le k+1$, define the translation vector $v_{\mu\nu} \coloneq e_\mu-e_\nu$ (the standard basis vectors), and define $\eta_{\mu\nu}\coloneq \frac{1}{2k}w_\mu+\frac{1}{2k}w_\nu$. Moreover, define the point \[M(W)\coloneq \left(\frac{1}{2k}+\frac{k-1}{2k}w_1,\frac{1}{2k}+\frac{k-1}{2k}w_2,\dots,\frac{1}{2k}+\frac{k-1}{2k}w_{k+1}\right) = \frac{k-1}{2k}W+\frac{k+1}{2k}M_\Delta\] where $M_\Delta \coloneq (\frac{1}{k+1},\frac{1}{k+1},\dots,\frac{1}{k+1})$ is the centroid of $\Delta^k$. Now, define \[K(W) \coloneq M(W)+\left\{\sum_{1\le \mu < \nu \le k+1}\xi_{\mu\nu}v_{\mu\nu}:\xi_{\mu\nu}\in[-\eta_{\mu\nu},\eta_{\mu\nu}]\text{ for all } 1 \le \mu < \nu \le k+1\right\}\] which is manifestly a zonotope.

\textbf{Claim I.} \textit{We have $(1-\lambda)W + \lambda \Delta^k \subseteq K(W)$ (i.e. the left containment of \eqref{eq:zonotope-bary}).}
\begin{proof}
    It suffices to check that each vertex of the lhs is contained in the rhs. By symmetry, it suffices to check only the first vertex of the lhs. Namely, we only need to verify $(1-\lambda)W + \lambda e_1 \in K(W)$.

    Recalling that $\lambda = \frac{1}{k}$, we compute \[(1-\lambda)W + \lambda e_1 - M(W) = \left(\frac{1}{2k}+\frac{k-1}{2k}w_1,-\frac{1}{2k}+\frac{k-1}{2k}w_2,-\frac{1}{2k}+\frac{k-1}{2k}w_3,\dots,-\frac{1}{2k}+\frac{k-1}{2k}w_{k+1}\right)\] which is equal to \[\begin{array}
        {rrrrcr}
        (\frac{k}{2k}w_1&+\frac{1}{2k}w_2&+\frac{1}{2k}w_3&+\frac{1}{2k}w_4&\cdots&+\frac{1}{2k}w_{k+1},\\
        -\frac{1}{2k}w_1&+\frac{k-2}{2k}w_2&-\frac{1}{2k}w_3&-\frac{1}{2k}w_4&\cdots&-\frac{1}{2k}w_{k+1},\\
        -\frac{1}{2k}w_1&-\frac{1}{2k}w_2&+\frac{k-2}{2k}w_3&-\frac{1}{2k}w_4&\cdots&-\frac{1}{2k}w_{k+1},\\
        -\frac{1}{2k}w_1&-\frac{1}{2k}w_2&-\frac{1}{2k}w_3&+\frac{k-2}{2k}w_4&\cdots&-\frac{1}{2k}w_{k+1},\\
        \vdots&\vdots&\vdots&\vdots&\ddots&\vdots\\
        -\frac{1}{2k}w_1&-\frac{1}{2k}w_2&-\frac{1}{2k}w_3&-\frac{1}{2k}w_4&\cdots&+\frac{k-2}{2k}w_{k+1})
    \end{array}\] which can be expressed in the form $\displaystyle\sum_{1\le \mu < \nu \le k+1}\xi_{\mu\nu}v_{\mu\nu}$ if we set \[\xi_{\mu\nu}=\begin{cases}
        \frac{1}{2k}w_\mu-\frac{1}{2k}w_\nu&\text{if } \mu>1\\
        \frac{1}{2k}w_\mu+\frac{1}{2k}w_\nu=\eta_{\mu\nu}&\text{if } \mu=1.
    \end{cases}\]
\end{proof}

\textbf{Claim II.} \textit{We have $K(W)\subseteq \Delta^k$ (i.e. the right containment of \eqref{eq:zonotope-bary}).}
\begin{proof}
    By symmetry, it suffices to verify that for all possible choices of $\xi_{\mu\nu}$, we have \[\frac{1}{2k}+\frac{k-1}{2k}w_1 + \sum_{\nu>1}\xi_{1\nu} \ge 0.\] This follows from the fact that \[\frac{1}{2k}+\frac{k-1}{2k}w_1 = \frac{k}{2k} w_1 + \frac{1}{2k}(w_2+w_3+\dots+w_{k+1})= \sum_{\nu>1}\eta_{1\nu}.\]
\end{proof}

Observe that we may also write \[K(W)=W+\left\{\sum_{1\le \mu < \nu \le k+1}\xi_{\mu\nu}v_{\mu\nu}:\xi_{\mu\nu}\in\left[-\frac{1}{k}w_\mu,\frac{1}{k}w_\nu\right]\text{ for all } 1 \le \mu < \nu \le k+1\right\}.\] 

The corresponding data are as follows: We have $r=\binom{k+1}{2}$, where we identify the set $\{1,2,\dots,r\}$ with the set of pairs $(\mu,\nu)$. Now let $e\in E$ be fixed. Set \[W=(w_1,w_2,\dots,w_{k+1})=(f_1(e),f_2(e),\dots,f_k(e),1-f_1(e)-f_2(e)-\dots-f_k(e)).\] Now set \begin{align*}
    l_{\mu\nu}(e) &= -\frac{1}{k}w_\mu\\
    u_{\mu\nu}(e) &= \frac{1}{k}w_\nu\\
    \alpha_{i\mu\nu}&= \begin{cases}
        1 &\text{if } i=\mu\\
        -1&\text{if } \nu<k+1 \text{ and } i=\nu\\
        0&\text{otherwise}
    \end{cases}\\
    \lambda &= \frac{1}{k}
\end{align*} As noted above, this data satisfies \eqref{eq:feas} and \eqref{eq:effic}. We may interpret the variables $\xi_{\mu\nu}$ as representing ``taking flow from commodity $\nu$ and giving flow to commodity $\mu$'' (the commodity $k+1$ representing ``unused capacity'' or ``the null commodity'').

\begin{proof}[Proof of Theorem \ref{thm:slow-flow}]
Let us set aside the problem of how do we obtain an initial feasible point to \eqref{eq:min-cost-flow}, and assume we currently have a feasible point $(f_1,f_2,\dots,f_k)$; now how do we make at least $\frac{1}{k}$ of progress?

The data above is a linear program in the form of Lemma \ref{lem:n25}, with $n' = kn$ and $m' = \binom{k+1}{2}m$. The matrix $A'$ is a block matrix, whose blocks are scalar multiples of the matrix $B$; furthermore, $A'$ is $4$-row-sparse. Thus, setting aside for now the problem of how do we obtain $x_\text{init}$, Lemma \ref{lem:n25} implies that we can solve \eqref{eq:zonotope} in time $\tilde{O}(\poly(k)(n^{2.5}+m\sqrt{n}))$. As noted previously, in order to solve \eqref{eq:min-cost-flow} we only need to solve $O(k\log n)$ zonotope linear programs; since each one takes time $\tilde{O}(\poly(k)(n^{2.5}+m\sqrt{n}))$, the total runtime is still $\tilde{O}(\poly(k)(n^{2.5}+m\sqrt{n}))$, as claimed in Theorem \ref{thm:slow-flow}.

How do we obtain $x_\text{init}$? Note that $0$ is guaranteed to be a feasible point, and in order to obtain the $1/\poly(n)$ separation required, we simply perturb $l$ and $u$ slightly. At the end this means we have solved a slightly different problem. We can perturb the solution so as to obtain an almost-feasible solution to the original problem.

Now we will say how to obtain an initial feasible point to \eqref{eq:min-cost-flow}. Create a new vertex $v^*$ and add edges between $v^*$ and every other vertex, with very large $\poly(n)$ capacities and costs. Now there is guaranteed to be a flow that satisfies all vertex demands.
\end{proof}

\textbf{Note.} It is not possible to give data satisfying $\lambda > \frac{1}{k}$ and also satisfying both \eqref{eq:feas} and \eqref{eq:effic}; in other words, our choice of $\lambda$ is optimal. To see this, simply observe that all zonotopes are, up to translation, centrally-symmetric convex bodies. Thus, \eqref{eq:zonotope-bary} would imply that $-\lambda \Delta^k + v \subseteq \Delta^k$ for some translation vector $v$, but this is not possible if $\lambda > \frac{1}{k}$. However, it may be possible that by using a smaller value of $\lambda$, one can reduce $r$, in such a way that the total runtime is reduced. 

\section{Solvers for zonotope flows}

%the literature is a mess

This section presents the general LP result used by the reduction as well as the additional structural speedup needed for the advertised running time.

\begin{lemma}\label{lem:n25}
    Let $m,n$ be positive integers satisfying $n \le m \le \poly(n)$, and let $c,l,u\in\RR^{m}$, $b\in\RR^{n}$, $A\in\RR^{m\times n}$ such that the magnitudes and condition numbers are $\le\poly(n)$. \textbf{Further assume that $A$ is $k$-row-sparse; i.e. each row of $A$ contains only $k$ nonzero entries.}
    
    Consider the linear program \[\begin{array}{rll}
    \text{Minimize}&c^\top x&\\
    \text{where}&x\in\RR^{m}&\text{are variables}\\
    \text{subject to:}&l(e)\le x(e)\le u(e)&\text{for all } 1 \le e \le m\\
    &A^\top x = b&
\end{array}\] Assume the linear program is feasible and $A$ has full column rank. Suppose we are given an initial point $x_\textnormal{init}$ which is a almost-feasible point, with margins at least $1/\poly(n)$. Then in time $\tilde{O}(n^{1/2}(km + k^2n^{2}))$ we can find $x$ which is almost-feasible such that $\|A^\top x - b\|_2 \le 1/\poly(n)$ and such that $c^\top x \le (1/\poly(n))+\textnormal{OPT}$. 
\end{lemma}

This is similar to Theorem 21 from \cite{LS15a}, but our proof is somewhat different. In particular, we do not have any dependence on $\omega$.

\begin{proof}[Proof of Lemma \ref{lem:n25}]
By the work in \cite{LS14b} (see \cite{LS20} for a more polished version), we know that in order to solve the linear program in Lemma \ref{lem:n25}, we can run an interior-point method that uses only $\tilde{O}(\sqrt{n})$ iterations, each iteration requiring us to solve a linear system of the form $A^\top W A x = y$ for some diagonal matrix $W$ (``weight matrix''). It suffices to give a solver for the matrix $A^\top W A$ that runs in time $\tilde{O}(km+k^2n^2)$. In what follows, we will assume $W$ is the identity matrix (the argument goes through for general $W$ mutatis mutandis).

Recall the recent result by Kelner, namely that we can solve a system of the form $Mx=y$ in time $\tilde{O}(n \cdot \nnz(M))$ (Theorem 1.1 of \cite{kelner2026solvinglinearsystemswidetildeomn}). If we naively use this to get a solver for the matrix $A^\top A$, this will take too long, since $A^\top A$ is in general a dense matrix. Instead, we will \emph{sparsify} $A$: sample $\tilde{O}(n)$ of its rows according to some probability distribution, resulting in the random matrix $\tilde{A}$, such that whp $\tilde{A}^\top \tilde{A}$ is a $2$-approximation to $A^\top A$ (in the sense of positive semi-definite matrices). The key idea is that $\tilde{A}^\top \tilde{A}$ is \emph{sparse}: it has only $\tilde{O}(k^2n)$ nonzero entries. Now we can apply Kelner's result; by preconditioned Richardson iteration, this gives a solver for the original matrix $A^\top A$ in time $\tilde{O}(km+k^2n^2)$ as desired.

How do we obtain $\tilde{A}$? We sample by \emph{leverage scores}.

Recall that if the rows of $A$ are labeled $a_1,a_2,\dots,a_m$, then the \textbf{$i^\text{th}$ leverage score} is given by \[\tau_i(A) \coloneq a_i^\top (A^\top A)^{-1} a_i\] which is a real number in the interval $[0,1]$. Observe that \[\sum_i \tau_i(A) = \tr\left(\sum_i a_i^\top (A^\top A)^{-1}a_i\right) = \tr\left(\sum_i a_ia_i^\top (A^\top A)^{-1}\right) = \tr((A^\top A)(A^\top A)^{-1}) = n.\] We have the following lemma, due to \cite{rudelson-vershynin} and \cite{Tropp_2011}:

\begin{lemma}\label{lem:lev-score-sample}
    Suppose $p_i\in [0,1]$ satisfy $p_i \ge O(\log n) \cdot \tau_i(A)$ for each $i$. Then, if we independently sample each row of $A$ with probability $p_i$, this will result in a matrix $\tilde{A}$ such that whp $\tilde{A}^\top \tilde{A}$ is a $2$-approximation to $A^\top A$. %We may round each $p_i$ up to the nearest multiple of $1/n$, in which case the sampling takes $\tilde{O}(m)$ time whp.
\end{lemma}

So if we could somehow obtain the leverage scores $\tau_i(A)$ (or even $\tilde{O}(1)$-approximations thereof), we could use Lemma \ref{lem:lev-score-sample} to obtain the desired sparsified $\tilde{A}$. The problem is that naively computing the leverage scores involves giving a solver for the matrix $A^\top A$, which is exactly the problem we started out with!

Instead, we follow the methods in \cite{lmp13}. For each $\lambda \ge 0$ define the auxiliary matrix $A_\lambda$ to consist of the square matrix $\lambda \cdot I_n$ appended below $A$. When $\lambda = 0$, $A_\lambda$ will have all the same leverage scores as $A$. Our strategy is as follows: \begin{itemize}
    \item Initialize $\lambda = \poly(n)$, large enough so that the leverage scores of $A_\lambda$ are easy to compute. Write $\tilde{\tau}_i(A_\lambda)$ for $2$-approximations to the leverage scores.
    \item Given $\tilde{\tau}_i(A_\lambda)$, use Lemma \ref{lem:lev-score-sample} to obtain a sparsified $\tilde{A}_\lambda$, with $\tilde{O}(n)$ rows, such that $\tilde{A}_\lambda^\top \tilde{A}_\lambda$ is a $2$-approximation to $A_\lambda^\top A_\lambda$. It follows that $\tilde{A}_\lambda^\top \tilde{A}_\lambda$ is a $8$-approximation to $A_{\lambda/2}^\top A_{\lambda/2}$.
    \item As noted above, we can use $\tilde{A}_\lambda^\top \tilde{A}_\lambda$ to do preconditioned Richardson iteration to obtain a solver for $A_{\lambda/2}^\top A_{\lambda/2}$ that runs in time $\tilde{O}(km+k^2n^2)$. Apply Lemma \ref{lem:gauss-solve} with this solver as the oracle; we see that we obtain $\tilde{\tau}_i(A_{\lambda/2})$ in time $\tilde{O}(km+k^2n^2)$.
    \item Set $\lambda \gets \lambda/2$ and repeat.
\end{itemize}

After running this loop $\tilde{O}(1)$ times, we will have $\lambda = 1/\poly(n)$, at which point $\lambda$ is so small that we can use $\tilde{\tau}_i(A_\lambda)$ as approximations of $\tau_i(A)$. As noted above, this allows us to use Lemma \ref{lem:lev-score-sample} to obtain $\tilde{A}$, as desired.
\end{proof}

\begin{lemma}\label{lem:gauss-solve}
    Suppose we are given an oracle solver for the matrix $A^\top A$. Then whp we can compute $\tilde{\tau}_i(A)$ which are $2$-approximations to the leverage scores $\tau_i(A)$ in time $\tilde{O}(\nnz(A))$ as well as $\tilde{O}(1)$ oracle calls.
\end{lemma}
\begin{proof}
    The key identity is \[\tau_i(A) = \|A(A^\top A)^{-1}a_i\|_2^2.\] If $G$ is a random Gaussian $O(\log n) \times m$ matrix, then whp \[\tau_i(A) \approx_2 \|GA(A^\top A)^{-1}a_i\|_2^2\] by Johnson-Lindenstrauss lemma. Our algorithm is as follows.

    First compute $G$ (this takes time $\tilde{O}(m)$). Now compute $GA$ (this takes time $\tilde{O}(\nnz(A))$). Now run the oracle solver, once for each column of $(GA)^\top$; since there are $O(\log n)$ columns, this takes $O(\log n)$ oracle calls and results in the matrix $Z \coloneq GA(A^\top A)^{-1}$.

    Finally we may simply compute $Za_i$ for each $i$; this takes time $\tilde{O}(\nnz(A))$. We take $\tilde{\tau}_i(A)$ to be $\|Za_i\|_2^2$.
\end{proof}

%\todo{talk about using the more recent LS19 version,
%and then add in bit complexity bounds from Jan and Mehrdad's more recent papers...}

%\todo{just some statement to convince reader this is likely is fine, this is 50 pages...}

%\textbf{Note.} Under the same conditions as Lemma \ref{lem:n25},
%    but with the linear program coming from a $r$-commodity zonotope flow (Definition~\ref{def:zflow}),
%    the runtime of Lemma \ref{lem:n25} improves to $\tilde{O}(\poly(r)(n^\omega+m))$. 
%For this, observe that the matrix $A$ is a block matrix with $r \times k$ many blocks ($r \geq k$),
%and each block is a scalar multiple of the edge-vertex incidence matrix $B$.
%Furthermore, for a set of block rows of this matrix,
%linearity gives that multiplying $A$ against a dual solution gives
%\[
%\sum_{i} \alpha_{ij} B y_{i}
%=
%B \left( \sum_{i} \alpha_{ij} y_{i} \right),
%\]
%which is the same as running a graph-theoretic
%heavy hitter on the vector $\sum_{i} \alpha_{ij} y_i$.
%So up to a factor of $r$, the
%expander decomposition based heavy-hitter,
%as given in~\cite{BLNPSS0W20} still works.
%Combining this heavy hitter with the linear programming algorithm from~\cite{BLLSSSW21}, we have the running time.

%The runtime of Theorem \ref{thm:slow-flow} would improve to $\tilde{O}(\poly(k)(n^\omega+m))$. 
\section{Acknowledgements}
CT thanks Yang P.~Liu and Richard Peng for explaining the literature pertaining to \cite{LS15a}.

%\bibliographystyle{alpha}
%\bibliography{refs}
\printbibliography

\end{document}